\documentclass{llncs}

\usepackage{xspace}
\usepackage{amsmath}
\usepackage{amsfonts}
\usepackage{listings}
\usepackage{fancyhdr}

\usepackage{mycommands}

\newcommand{\fixC}{\ensuremath{\mathit{K\hspace{-0.3em}a}}}
\newcommand{\C}{\ensuremath{\mathit{C\hspace{-0.3em}a}}}
\newcommand{\Cex}{\ensuremath{\mathit{C\hspace{-0.3em}A}}}
\newcommand{\ren}{\ensuremath{\mathsf{ren}}}
\newcommand{\rep}{\ensuremath{\mathfrak{A}}}
\newcommand{\anchor}{\ensuremath{\mathbf{a}}\xspace}
\newcommand{\cXVn}{\ensuremath{f}}
\newcommand{\ctxtstmt}{\ensuremath{\mathsf{St}}}
\newcommand{\ctxtfunc}{\ensuremath{\mathsf{Cx}}}
\newcommand{\cxvn}{contextualization\xspace}
\newcommand{\annstmt}[2]{\ensuremath{\langle#1,#2\rangle}}

\newcommand{\dl}{Description Logic\xspace}
\newcommand{\interp}{\mathcal{I}}
\newcommand{\nonctxtterms}{\ensuremath{\mathsf{N}^{\mathsf{nc}}}}
\newcommand{\ctxtterms}{\ensuremath{\mathsf{N}^{\mathsf{c}}}}
\newcommand{\ctxt}{\ensuremath{\mathsf{N}^{\mathsf{a}}}}
\newcommand{\terms}{\mathsf{N}}
\newcommand{\ctxtann}{\ensuremath{\mathfrak{C}}}
\newcommand{\icpo}{\ensuremath{\ct{isContextualPartOf}}}
\newcommand{\iic}{\ensuremath{\ct{isInContext}}}
\newcommand{\Nd}{\textit{NdTerms}\xspace}
\newcommand{\rel}{\ensuremath{\mathsf{Rel}}}
\newcommand{\relf}{\ensuremath{\mathsf{rel}}}
\newcommand{\model}{\ensuremath{\langle\Delta^\interp,\cdot^\interp\rangle}}
\newcommand{\sig}{\mathsf{Sig}}
\newcommand{\ax}{\mathsf{Ax}}
\newcommand{\ct}[1]{\ensuremath{\mathtt{#1}}}

\begin{document}

\title{Integrating Context of Statements within Description Logics}
\titlerunning{NdProperties}  
%
\author{Antoine Zimmermann\inst{1} \and Jos\'e M. Gim\'enez-Garc\'ia\inst{2}}
\authorrunning{Gim\'enez-Garc\'ia and Zimmermann} 
%
\tocauthor{Antoine Zimmermann and Jos\'e M. Gim\'enez-Garc\'ia}
\institute{Univ Lyon, MINES Saint-\'Etienne, CNRS, Laboratoire Hubert Curien\\ UMR 5516, F-42023 Saint-\'Etienne, France\\
\email{antoine.zimmermann@emse.fr}
\and
Universit\'e de Lyon, CNRS, UMR 5516,\\ Laboratoire Hubert-Curien, Saint-\'Etienne, France\\
\email{jose.gimenez.garcia@univ-st-etienne.fr}
}

\maketitle              

\thispagestyle{fancy}

\begin{abstract}
We address the problem of providing contextual information about a logical formula (\eg provenance, date of validity, or confidence) and representing it within a logical system. In this case, it is needed to rely on a higher order or non standard formalism, or some kind of reification mechanism. We explore the case of reification and formalize the concept of contextualizing logical statements in the case of Description Logics. Then, we define several properties of contextualization that are desirable. No previous approaches satisfy all of the them.
Consequently, we define a new way of contextually annotating statements. It is inspired by NdFluents, which is itself an extension of the 4dFluents approach for annotating statements with temporal context. In NdFluents, instances that are involved in a contextual statement are sliced into contextual parts, such that only parts in the same context hold relations to one another, with the goal of better preserving inferences. We generalize this idea by defining contextual parts of relations and classes. This formal construction better satisfies the properties, although not entirely. We show that it is a particular case of a general mechanism that NdFluents also instantiates, and present other variations. 

\end{abstract}

\section{Introduction} \label{sec:intro}

The problem of being able to reason not only \emph{with} logical formulas, but also \emph{about} said formulas, is an old one in artificial intelligence. 
McCarthy~\cite{McCarthy:87} proposed to extend first order logic by reifying context and formulas to introduce a binary predicate $\mathsf{ist}(\phi,c)$ satisfied if the formula $\phi$ is true ($\mathsf{ist}$) in the context $c$. However, a complete axiomatization and calculus for McCarthy's contextual logic has never been formalized. Giunchiglia~\cite{Giunchiglia:93} proposed
the grouping of ``local'' formulas in contexts, and then using other kinds of formulas to characterize how knowledge from multiple contexts is \emph{compatible}.
This idea of locality+compatibility \cite{Ghidini_Serafini:00} has led to several non standard formalisms for reasoning with multiple contexts~\cite{Zimmermann:13}.
Alternatively, the approach of annotated logic programming~\cite{Kifer_Subrahmanian:92} considers that a contextual annotation is just a value in an algebraic structure (\eg a number or a temporal interval). 
This idea was later applied to annotated RDF and RDFS~\cite{Udrea_Reforgiato_Subrahmanian:10,Zimmermann_Lopes_Polleres_Straccia:12}. 

The representation of statement annotation has sometimes being thought of as a data model problem without consideration of the logical formalism behind. In particular, several proposals to extend the RDF data model in various ways for allowing annotations have been made: named graphs~\cite{Carroll_Bizer_Hayes_Stickler:05}, RDF+~\cite{Dividino_Sizov_Staab_Schueler:09}, RDF*~\cite{Hartig_Thompson:14}, Yago Model~\cite{Hoffart_Suchanek_Berbeich_Weikum:13}. However, the underlying data structures have not a clear formal semantics.
Therefore, some authors advocate another approach to representing annotation of knowledge: reify the statement or its context and describe it within the formalism of the statement. This requires modifying the statement so as to integrate knowledge of the context or statement. Examples of such techniques are reification~\cite{RDF_Schema_W3C:14}, N-Ary Relations~\cite{N-ary_Relations_W3C:06}, Singleton Property~\cite{Nguyen_Bodenreider_Sheth:14}), and NdFluents~\cite{GimenezGarcia_Zimmermann_Maret:17}. \emph{This paper provides an abstraction of the reification techniques in the context of Description Logics (DLs) in the form of what we call contextualization functions}. Additionally, we introduce a new technique for the representation of contextual annotations that satisfies better some desirable properties.


After introducing our notations for DLs in \refsec{preliminaries}, we provide formal definitions that allow us to define verifiable properties of the reification techniques (\refsec{context}). Our new technique, named \Nd, is presented in \refsec{nd.approach}, where we also prove to what extent it satisfies the properties of the previous section. \refsec{multiple.domains} discuss some of the problems that may occur when combining knowledge having different annotations. In \refsec{other.approaches}, we present how the other approaches fit in our formalization and why they do not satisfy well the properties. Finally, we discuss this and future work in \refsec{discussion}.

\section{Preliminaries}\label{sec:preliminaries}

In this section, we introduce the notations and definitions we use in relation to Description Logics. Note that we use an extended version of DL where all terms  can be used as \emph{concept names}, \emph{role names}, and \emph{individual names} in the same ontology. Using the same name for different types of terms is known as ``punning'' in OWL~2~\cite[Section~2.4.1]{OWL_New_Features_W3C:12}. Moreover, we allow more constructs than in OWL~2 DL and make no restriction on their use in order to show that our approach is not limited to a specific DL.

We assume that there is an infinite set of \emph{terms} $\terms$. Every term is an \emph{individual}. A \emph{role} is either a term or, given roles $R$ and $S$, $R\sqcup S$, $R\sqcap S$, $\neg R$, $R^-$, $R\circ S$ and $R^+$. A \emph{concept} is either a term, or, given concepts $C$, $D$, role $R$, individuals $u_1,\dots,u_k$, and natural number $n$, $\bot$, $\top$, $C\sqcup D$, $C\sqcap D$, $\exists R.C$, $\forall R.C$, $\leq n R.C$, $\geq nR.C$, $\neg C$ or $\{u_1,\dots,u_k\}$. Finally, we also allow concept product $C\times D$ to define a role.

Interpretations are tuples $\langle\Delta^\interp,\cdot^{\interp_u},\cdot^{\interp_r},\cdot^{\interp_c}\rangle$, where $\Delta^\interp$ is a non-empty set (the domain of interpretation) and $\cdot^{\interp_u}$, $\cdot^{\interp_r}$, and $\cdot^{\interp_c}$ are the interpretation functions for individuals, roles and concepts respectively such that:
\begin{itemize}
 \item for all $u\in\terms$, $u^{\interp_u}\in\Delta^\interp$;
 \item for all $P\in\terms$, $P^{\interp_r}\subseteq\Delta^\interp\times\Delta^\interp$ and interpretation of roles is inductively defined by $(R\sqcup S)^{\interp_r}=R^{\interp_r}\cup S^{\interp_r}$, $(R\sqcap S)^{\interp_r}=R^{\interp_r}\cap S^{\interp_r}$, $(\neg R)^{\interp_r}=(\Delta^\interp\times\Delta^\interp)\setminus R^{\interp_r}$, $(R^-)^{\interp_r}=\{\langle x,y\rangle\vert\langle y,x\rangle\!\in\! R^{\interp_r}\}$, $(R\circ S)^{\interp_r}=\{\langle x,y\rangle\vert\exists z.\langle x,z\rangle\!\in\! R^{\interp_r}\wedge\langle z,y\rangle\!\in\! S^{\interp_r}\}$ and $(R^+)^{\interp_r}$ is the reflexive-transitive closure of $R^{\interp_r}$ (with $R$ and $S$ being arbitrary roles).
  \item for all $A\in\terms$, $A^{\interp_c}\subseteq\Delta^\interp$ and interpretation of concepts is defined by $\bot^{\interp_c}=\emptyset$, $\top^{\interp_c}=\Delta^\interp$, $(C\sqcup D)^{\interp_c}=C^{\interp_c}\cup D^{\interp_c}$, $(C\sqcap D)^{\interp_c}=C^{\interp_c}\cap D^{\interp_c}$, $(\exists R.C)^{\interp_c}=\{x\vert\exists y.y\!\in\! C^{\interp_c}\wedge\langle x,y\rangle\!\in\! R^{\interp_r}\}$, $(\forall R.C)^{\interp_c}=\{x\vert\forall y.\langle x,y\rangle\!\in\! R^{\interp_r}\Rightarrow y\!\in\! C^{\interp_c}\}$, $(\leq nR.C)^{\interp_c}=\{x\vert\sharp\{y\!\in\! C^{\interp_c}\vert\langle x,y\rangle\!\in\! R^{\interp_r}\}\leq n\}$, $(\geq nR.C)^{\interp_c}=\{x\vert\sharp\{y\!\in\! C^{\interp_c}\vert\langle x,y\rangle\!\in\! R^{\interp_r}\}\geq n\}$, $(\neg C)^{\interp_c}=\Delta^\interp\setminus C^{\interp_c}$, $\{u_1,\dots,u_k\}=\{u_1^{\interp_u},\dots,u_k^{\interp_u}\}$, where $C$ and $D$ are arbitrary concepts, $R$ an arbitrary role, $u_1,\dots,u_k$ are individual names, and $k$ and $n$ two natural numbers.
   \item Roles defined as a concept product are interpreted as $(C\times D)^{\interp_r}=C^{\interp_c}\times D^{\interp_c}$ for arbitrary concepts $C$ and $D$.
\end{itemize}

In the following, we slightly abuse notations by defining interpretations as pairs $\model$ where $\cdot^\interp$ denotes the three functions $\cdot^{\interp_u}$, $\cdot^{\interp_c}$, and $\cdot^{\interp_r}$. Moreover, when we write $x^\interp=y^{\interp'}$, it means ``$x^{\interp_u}=y^{\interp'_u}$ and $x^{\interp_c}=y^{\interp'_c}$ and $x^{\interp_r}=y^{\interp'_r}$''.

\emph{Axioms} are either general concept inclusions $C\sqsubseteq_c D$, sub-role axioms $R\sqsubseteq_r S$, instance assertions $C(a)$, or role assertions $R(a,b)$, where $C$ and $D$ are concepts, $R$ and $S$ are roles, and $a$ and $b$ are individual names.
An interpretation $\interp$ \emph{satisfies} axiom $C\sqsubseteq_c D$ \iff $C^{\interp_c}\subseteq D^{\interp_c}$; it satisfies $R\sqsubseteq_r S$ \iff $R^{\interp_r}\subseteq S^{\interp_r}$; it satisfies $C(a)$ \iff $a^{\interp_u}\!\in\! C^{\interp_c}$; and it satisfies $R(a,b)$ \iff $\langle a^{\interp_u},b^{\interp_u}\rangle\!\in\! R^{\interp_r}$. When $\interp$ satisfies an axiom $\alpha$, it is denoted by $\interp\models\alpha$. Instance assertions, role assertions and individual identities constitute the ABox axioms.

An ontology $O$ is composed of a set of terms called the signature of $O$ and denoted by $\sig(O)$, and a set of axioms denoted by $\ax(O)$. An interpretation $\interp$ is a model of an ontology $O$ \iff for all $\alpha\in\ax(O)$, $\interp\models\alpha$. In this case, we write $\interp\models O$. The set of all models of an ontology $O$ is denoted by $\mathrm{Mod}(O)$.
A \emph{semantic consequence} of an ontology $O$ is a formula $\alpha$ such that for all $\interp\in\mathrm{Mod}(O)$, $\interp\models\alpha$.

In the rest of the paper, we will use \texttt{teletype font} to denote known individuals, and normal font for unknown individuals and variables (\eg $\ct{City}(\ct{babylon})$ and $\ct{City}(x)$).

\section{Contextualization of Statements}\label{sec:context}





A contextual annotation can be thought of as a set of ABox axioms that describe an individual representing the statement (the anchor) that is annotated. An annotated statement (or ontology) is the combination of a DL axiom (or DL ontology) with a contextual annotation.

\begin{definition}[Connected individuals]\label{def:connected.individuals}
Two terms $a$ and $b$ are \emph{connected individuals} \wrt an ABox $A$ \iff $a$ and $b$ are used as individual names in $A$, and either
\begin{itemize}
	\item $a$ and $b$ are the same term, or
    \item there exists $R_1,\cdots,R_n$ and $z_1,\cdots,z_{n-1}$, such that:
    \begin{itemize}
    	\item $R_1(a,z_1)$, or $R_1(z_1,a)$
        \item $R_i(z_{i-1},z_i)$, or $R_i(z_i,z_{i-1})$, $2\leq i \leq n-2$
        \item $R_n(z_{n-1},b)$, or $R_n(b,z_{n-1})$
    \end{itemize}
\end{itemize}
\end{definition}

\begin{example}\label{ex:connected.individuals}
If we consider the ABox $A = \{P(a,b),$ $Q(c,b),$ $S(d,e)\}$, the pairs of individuals $\{a,b\}$, $\{b,c\}$, $\{a,c\}$, and $\{d,e\}$ are \emph{connected individuals}, but $\{a,d\}$, $\{b,d\}$, $\{c,d\}$, $\{a,e\}$, $\{b,e\}$, and $\{c,e\}$ are not.
\end{example}

\begin{definition}[Contextual annotation]\label{def:contextual.annotation}
A \emph{contextual annotation} $\C$ is an ABox with signature $\{\anchor\}\cup\Sigma$ where $\anchor\not\in\Sigma$ is a distinguished term (called the \emph{anchor}) and $\Sigma$ is a DL signature such that $\forall x \in \Sigma$, $\{a,x\}$ are \emph{connected individuals}.
\end{definition}

\begin{example}\label{ex:contextual.annotation}
The Abox $\Cex=\{\ct{validity}(\anchor,\ct{t}),$ $\ct{Interval}(\ct{t})$ $\ct{from}(\ct{t},\ct{609BC}),$ $\ct{to}(\ct{t},\ct{539BC})$, $\ct{prov}(\anchor,\ct{w}),$ $\ct{name}(\ct{w},\ct{wikipedia}),$ $\ct{Wiki}(\ct{w})\}$ is a \emph{contextual annotation}, where $\anchor$ is the anchor and $\Sigma = \{\ct{t},$ $\ct{Interval},$ $\ct{w},$ $\ct{wikipedia},$ $\ct{wiki},$ $\ct{609BC},$ $\ct{539BC}\}$.
\end{example}

\begin{definition}[Annotated statement]\label{def:annotated.statement}
An \emph{annotated statement} is a pair $\annstmt{\alpha}{\C}$ such that $\alpha$ is a description logic axiom and $\C$ is a contextual annotation.
\end{definition}

\begin{example}\label{ex:annotated.statement}
The pair $\annstmt{\alpha}{\Cex}$, where $\alpha = \ct{capital}(\ct{babylon},$ $\ct{babylonianEmpire})$ and $\Cex$ is the contextual annotation from \refex{contextual.annotation}, is an \emph{annotated statement}.
\end{example}

\begin{definition}[Annotated ontology]\label{def:annotated.ontology}
An \emph{annotated ontology} is a pair $\annstmt{O}{\C}$ such that $O$ is a description logic ontology and $\C$ is a contextual annotation.
\end{definition}


Each reification technique has an implicit construction plan in order to map an annotated statement to a resulting ontology. A contextualization (\refdef{contextualization}) represents the procedure that generates a single DL ontology from a given annotated statement or ontology. The procedure must not lose information, especially not the annotation. 


\begin{definition}[Contextualization]\label{def:contextualization}
A \emph{\cxvn} is a function $\cXVn$ that maps each annotated statement $\alpha_\C = \annstmt{\alpha}{\C}$ to a description logic ontology $\cXVn(\alpha_\C)$ = $\ctxtstmt(\alpha_\C) \cup \ctxtfunc(\alpha_\C)$ such that:
\begin{itemize}
	\item there exists an individual $u$ in the signature of $\ctxtstmt(\alpha_\C)$ and of $\ctxtfunc(\alpha_\C)$ such that:
    \begin{itemize}
    	\item for all $R(\anchor,x)\in \C$, $R(u,x)\in \ctxtfunc(\alpha_\C)$;
    	\item for all $R(x,\anchor)\in \C$, $R(x,u)\in \ctxtfunc(\alpha_\C)$;
    	\item for all $C(\anchor)\in \C$, $C(u)\in \ctxtfunc(\alpha_\C)$;
    	\item for all other $\alpha\in \C$, $\alpha\in \ctxtfunc(\alpha_\C)$.
  \end{itemize}
  	\item there is an injective mapping between the signature of $\alpha$ and the signature of $\ctxtstmt(\alpha_C)$. 
\end{itemize}
We extend $\cXVn$ to all annotated ontologies $O_\C = \annstmt{O}{\C}$ by defining $\cXVn(O_\C) = \bigcup_{\alpha\in O}\cXVn(\annstmt{\alpha}{\C})$.
\end{definition}

\begin{example}\label{ex:contextualization}
An example contextualization function $f_{ex}$ introduces a fresh term $t$ for each annotated statement with a role assertion $R(a,b)$, where $R$, $a$, and $b$ are three therms, creates new axioms $\ct{subject}(t,s)$, $\ct{predicate}(t,R)$, $\ct{object}(t,o)$, and finally removes the axiom $R(a,b)$. Notice that this construction requires the punning of term $R$. This function is analogous to \emph{RDF Reification}. The result of this contextualization, along with Other possible known approaches, is described in \refsec{other.approaches}.
\end{example}

Those are the only structures that we will consider in this paper. The remaining definitions are desirable properties that a contextualization should satisfy, especially if one wants it to preserve as much of the original knowledge as possible. 

\begin{definition}[Soundness]\label{def:soundness}
A \cxvn function $\cXVn$ is \emph{sound} \wrt a set of annotated ontologies $\Omega$ \iff for each $O_\C = \annstmt{O}{\C}\in\Omega$ such that $O$ and $\C$ are consistent, then $\cXVn(O_\C)$ is consistent.
\end{definition}

That is, a contextualization is sound if, when contextualizing a consistent ontology, the result is also consistent. This property avoids that the contextualization introduces unnecessary contradictions that would result in everything being entailed by it. Note that this requirement is not necessary in the opposite direction, \ie if $\cXVn(O_\C)$ is consistent, it is not required that $O$ and $C$ are consistent.  

\begin{example}\label{ex:soundness}
The contextualization function $f_{ex}$ from \refex{contextualization} is sound \wrt the set of ontologies $\Omega$, where $\Omega \cup \{\ct{subject},$ $\ct{predicate},$ $\ct{predicate}\} = \emptyset$.
\end{example}


\begin{definition}[Inconsistency preservation]\label{def:inconsistency.preservation}
Let $\cXVn$ be a \cxvn function. We say that $\cXVn$ preserves inconsistencies \iff for all annotated ontologies $O_\C = \annstmt{O}{\C}$, if $O$ is inconsistent then $\cXVn(O_\C)$ is inconsistent.
\end{definition}


Inconsistency preservation means that a self-contradictory ontology in a given context is contextualized into an inconsistent ontology, such that bringing additional knowledge from other contexts would result in no more consistency. If something is inconsistent within a context, then it is not really worth to consider reasoning with this annotated ontology.

\begin{example}\label{ex:inconsistency.preservation}
The contextualization function $f_{ex}$ from \refex{contextualization} does not preserve inconsistencies. For instance, $\ct{capitalOf}$ can be defined as irreflexive using the following axiom: $\exists \ct{capitalOf}.\top  \sqsubseteq \forall \ct{capitalOf^-}.\bot$. Then, the axiom $\ct{capitalOf}(\ct{babylon},$ $\ct{babylon})$ would make the ontology inconsistent. But when applying $f_{ex}$ the result is consistent.
 \end{example}

\begin{definition}[Entailment preservation]\label{def:entailment.preservation}
Let $\cXVn$ be a \cxvn function. Given two description logic ontologies $O_1$ and $O_2$ such that $O_1\models O_2$, we say that $\cXVn$ \emph{preserves the entailment} between $O_1$ and $O_2$ \iff for all contextual annotations $\C$, $\cXVn(\annstmt{O_1}{\C}) \models \cXVn(\annstmt{O_2}{\C})$.
Given a set $\mathfrak{K}\subseteq\ctxt$ of contextual annotations, if $\cXVn$ preserves all entailments between ontologies in $\mathfrak{K}$, then we say that $\cXVn$ is \emph{entailment preserving} for $\mathfrak{K}$.
\end{definition}

In short, a contextualization is entailment preserving if all the knowledge that could be inferred from the original ontology can also be inferred, in the same context, in the contextualized ontology.


\begin{example}\label{ex:entailment.preservation}
The contextualization function $f_{ex}$ from \refex{contextualization} preserves entailments for the TBox of ontologies (because no modifications are made on its axioms), but it does not preserve entailments on role assertions. For instance,  the axioms $\ct{capitalOf} \sqsubseteq \ct{cityOf}$, $\ct{capitalOf}(\ct{babylon},$ $\ct{babylonianEmpire})$ entails $\ct{cityOf}(\ct{babylon},$ $\ct{babylonianEmpire})$, but this inference is not preserved after applying $f_{ex}$.
\end{example}






\section{The \Nd Approach} \label{sec:nd.approach}

This section defines the \Nd approach that extends the NdFluent proposal~\cite{GimenezGarcia_Zimmermann_Maret:17}
. To this end, we assume that terms are divided into three infinite disjoint sets $\nonctxtterms$, $\ctxtterms$, and $\ctxt$ called the \emph{non contextual terms}, the \emph{contextual terms}, and the \emph{anchor terms} respectively.
We also assume that there is an injective function $\rep:\ctxtann\to\ctxt$ and for all $\C\in\ctxtann$ there is an injective function $\ren_\C:\nonctxtterms\to\ctxtterms$ and two terms $\icpo,\iic\in\nonctxtterms$. For any $\C$, we extend $\ren_\C$ to axioms by defining $\ren_\C(\alpha)$ as the axiom built from $\alpha$ by replacing all terms $t\in\sig(\alpha)$ with $\ren_\C(t)$.

\subsection{Contextualization Function in \Nd}\label{sec:contextualization.function.in.nd}

The contextualization needs to combine the ontologies from the statements and the contextual annotation. However, if we na\"ively make the union of the axioms, they could contradict, and it would not be possible to ensure the desired properties. For example, an ontology may restrict the size of the domain of interpretation to be of a fixed cardinality, while the contextual annotation may rely on more elements outside the local universe of this context. For this reason we use the concept or relativization: The ontology is modified in such a way that the interpretation of everything explicitly described in it is confined to a set, while
external terms or constructs may have elements outside said set. Relativization has been applied in various logical settings over the past four decades (\eg~\cite{Scott:79}) and applied to DLs and OWL~\cite{CuencaGrau_Kutz:07}, among others.

\begin{definition}[Relativization]\label{def:relativization}
Let $\C$ be a contextual annotation. Given an ontology $O$, the relativization of $O$ in $\C$ is an ontology $\rel_\C(O)$ built from $O$ as follows:
 \begin{enumerate}
  \item $\sig(\rel_\C(O)) = \sig(O) \cup \top_\C$ where $\top_\C$ is a term not appearing in $\sig(O)$;
  \item\label{itm:top} for all appearances of $\top$ in an axiom of $O$, replace $\top$ with $\top_\C$;
  \item\label{itm:notc} for all concepts $\neg C$ appearing in an axiom of $O$, replace it with $\neg C\sqcap\top_\C$;
  \item\label{itm:notr} for all roles $\neg R$ appearing in an axiom of $O$, replace it with $\neg R\sqcap (\top_\C\times\top_\C)$;
  \item\label{itm:forall} for all concepts $\forall R.C$ appearing in an axiom of $O$, replace it with $\forall R.C\sqcap\top_\C$;
  \item\label{itm:transclosure} for all roles $R^+$ appearing in an axiom of $O$, replace it with $R^+\sqcap\top_\C\times\top_\C$.
  \item\label{itm:extra} Additionally, for all terms $t\in\sig(O)$, the following axioms are in $\rel_\C(O)$:
   \begin{itemize}
    \item $t\sqsubseteq \top_\C$,
    \item $\top_\C(t)$,
    \item $\exists t.\top\sqsubseteq \top_\C$,
    \item $\top\sqsubseteq \forall t.\top_\C$.
   \end{itemize}
 \end{enumerate}
\end{definition}

The relativization of an ontology can be done systematically by relativizing its concepts and roles, which in turn can be achieved by using \refdef{reativization.of.concepts.and.roles}.

\begin{definition}[Relativization of concepts and roles]\label{def:reativization.of.concepts.and.roles}
Given a contextual annotation $\C$, we define a function $\relf_\C$ that maps concepts and roles to concepts and roles according to the rules of Items~\ref{itm:top}-\ref{itm:transclosure}. That is, recursively:
 \begin{itemize}
  \item $\relf_\C(t) = t$;
  \item $\relf_\C(\{u_1,\dots u_k\}) = \{\relf_\C(u_1),\dots relf(u_k)\}$;
  \item $\relf_\C(C \sqcup D) = \relf_\C(C) \sqcup \relf_\C(D)$;
  \item $\relf_\C(C \sqcap D) = \relf_\C(C) \sqcap \relf_\C(D)$;
  \item $\relf_\C(\neg C) = \neg\relf_\C(C)\sqcap\top_\C$;
  \item $\relf_\C(C\times D) = \relf_\C(C)\times\relf_\C(D)$;
  \item $\relf_\C(R \sqcup S) = \relf_\C(R) \sqcup \relf_\C(S)$;
  \item $\relf_\C(R \sqcap S) = \relf_\C(R) \sqcap \relf_\C(S)$;
  \item $\relf_\C(R \circ S) = \relf_\C(R) \circ \relf_\C(S)$;
  \item $\relf_\C(\neg R) = \neg\relf_\C(R)\sqcap\top_\C\times\top_\C$;
  \item $\relf_\C(R^-) = \relf_\C(R)^-$;
  \item $\relf_\C(R^+) = \relf_\C(R)^+\sqcap\top_\C\times\top_\C $;
  \item $\relf_\C(\exists R.C) = \exists\relf_\C(R).\relf_\C(C)$;
  \item $\relf_\C(\forall R.C) = \forall\relf_\C(R).\relf_\C(C) \sqcap \top_\C$;
  \item $\relf_\C(\geq n R.C) = \geq n~\relf_\C(R).\relf_\C(C)$;
  \item $\relf_\C(\leq n R.C) = \leq n~\relf_\C(R).\relf_\C(C)$.
 \end{itemize}
 where $t$ is a term, $C, D$ are concepts, $R, S$ are roles, $u_1,\dots u_k$ are individuals, and $k, n$ are natural numbers.
\end{definition}

\begin{example}
The axiom $\exists \ct{capitalOf}.\top  \sqsubseteq \forall \ct{capitalOf^-}.\bot$ from \refex{inconsistency.preservation} is relativized into $\exists \ct{capitalOf}.\top_\C \sqsubseteq \forall \ct{capitalOf^-}.\bot \sqcap \top_\C$.
\end{example}

Then, the contextualization in \Nd is done by: (1) creating the replacement of the anchor using the function \rep, (2) renaming all the terms in the statement using the \ren~function, (3) linking them to the original terms by the $\icpo$ relation, and (4) linking the renamed terms to the context using the $\iic$ relation.

\begin{definition}[Contextualization in \Nd]\label{def:contextualization.in.nd}
Let $\C\in\ctxtann$ be any contextual annotation. Let $\alpha_{\C} = \annstmt{\alpha}{\C}$ be an annotated statement such that the signatures of $\alpha$ and $\C$ are in $\nonctxtterms$. We define the contextualization function $\cXVn_{\mathsf{nd}}$ such that $\cXVn_{\mathsf{nd}}(\alpha_\C)$ = $\ctxtstmt(\alpha_\C) \cup \ctxtfunc(\C)$ and:
\begin{itemize}
 \item $\ctxtstmt_\C(\alpha) = \{\ren_\C(\rel_\C(\alpha))\} \cup \{\icpo(\ren_\C(t),t)\mid t\in\sig(\alpha)\} \cup \{\iic(\ren_\C(t),\rep(\C))\mid t\in\sig(\alpha)\}$.
 \item $\ctxtfunc(\C)$ contains exactly the following axioms:
  \begin{itemize}
   \item for all $R(\anchor,x)\in \C$, $R(\rep(\C),x)\in \ctxtfunc(\alpha)$;
   \item for all $R(x,\anchor)\in \C$, $R(x,\rep(\C))\in \ctxtfunc(\C)$;
   \item for all $C(\anchor)\in \C$, $C(\rep(\C))\in \ctxtfunc(\C)$;
   \item for all other axioms $\beta\in\C$, $\beta\in \ctxtfunc(\C)$.
  \end{itemize}
\end{itemize}
\end{definition}

Similarly to \refex{contextualization}, this construction requires punning, since all terms in the statement are used as individual names in the role assertion $\icpo(\ren_\C(t),$ $t)$.

\begin{example}
The NdTerms contextualization of our running example within the context $\Cex$ of \refex{contextual.annotation} contains the following axioms, where \emph{term@\C} is the result of the renaming function $\ren_\C($\emph{term}$)$:

$\ct{capitalOf@}(\ct{babylon@\Cex},\ct{babylonianEmpire@\Cex})$

$\icpo(\ct{babylon@\Cex},\ct{babylon})$

$\icpo(\ct{babylonianEmpire@\Cex},\ct{babylonianEmpire})$

$\iic(\ct{babylon},\ct{exampleContext})$

$\iic(\ct{babylonianEmpire},\ct{exampleContext})$

$\ct{validity}(\ct{exampleContext},\ct{t})$

$\ct{Interval}(\ct{t})$

$\ct{from}(\ct{t},\ct{609BC})$

$\ct{to}(\ct{t},\ct{539BC})$

$\ct{prov}(\ct{exampleContext,\ct{w}}$

$\ct{name}(\ct{w},\ct{wikipedia})$

$\ct{Wiki}(\ct{w})$
\end{example}

\subsection{Soundness of \Nd}\label{sec:soundness.of.nd}

In this section, we fix a contextual annotation $\fixC$ that has its signature in $\nonctxtterms$, since the following theorems and proofs require such a constraint.

The contextualization of \Nd is sound, but only \wrt annotated ontologies that satisfy certain conditions. In order to present the conditions, we need to introduce the following definition, that is also used in several proofs of this paper.

\begin{definition}[Domain extensibility]\label{def:domain.extensibility}
Let $O$ be an ontology. A model $\interp=\model$ of $O$ is \emph{domain extensible} for $O$ \iff for all sets $\Delta^+$, $\interp'=\langle\Delta^\interp\cup\Delta^+,\cdot^\interp\rangle$ is also a model of $O$. An ontology is said to be \emph{model extensible} \iff it has a model that is domain extensible.
\end{definition}

Note that, even if the domain of interpretation of an ontology is infinite, that does not necessarily mean that its models are domain extensible. This notion is closely related to the notion of \emph{expansion} in \cite{CuencaGrau_Kutz:07} since if $\interp$ is domain extensible, then one can build infinitely many expansions of it. 

\begin{theorem}[Soundness of \Nd]\label{thm:soundness.of.nd}
If the contextual annotation $\fixC$ is model extensible, then the contextualization function $\cXVn_{\mathsf{nd}}$ is sound \wrt annotated ontologies $O_\fixC = \langle O,\fixC\rangle$, $\sig(O)\subseteq\nonctxtterms$, and $\sig(O)\cap\sig(\fixC)=\emptyset$.
\end{theorem}

The proof of this theorem requires a few intermediary steps. \refthm{model.extensibility} ensures that, given the right condition, a model of the union of two ontologies can be made from two models of the original ontologies.

\begin{theorem}[Model extensibility theorem]\label{thm:model.extensibility}
Let $O$ and $O^+$ be two ontologies such that there exist two models $\interp=\model$ and $\interp^+\langle\Delta^{\interp^+},\cdot^{\interp^+}\rangle$ that are domain extensible for $O$ and $O^+$ respectively. If $\sig(O)\cap\sig(O^+)=\emptyset$ then $\interp$ and $\interp^+$ can be both extended into a unique model $\interp'$ of $O\cup O^+$ such that $\Delta^{\interp'}=\Delta^{\interp}\cup\Delta^{\interp^+}$ and for all $t\in\sig(O)$, $t^{\interp'}=t^\interp$, and for all $t\in\sig(O^+)$, $t^{\interp'}=t^{\interp^+}$.\footnote{Remember from preliminaries that $t^{\interp'}=t^{\interp}$ means ``$t^{\interp'_u}=t^{\interp_u}$ and $t^{\interp'_c}=t^{\interp_c}$ and $t^{\interp'_r}=t^{\interp_r}$''.}
\end{theorem}


\begin{proof}[\refthm{model.extensibility}]
Since $O^+$ is consistent, there exists a model $\interp^+$ of $O^+$ and \wolog, we assume that $\Delta^\interp\cap\Delta^{\interp^+}=\emptyset$. We define a new interpretation $\interp'$ where $\Delta^{\interp'}=\Delta^\interp\cup\Delta^+$; for all $t\in\sig(O)$, $t^{\interp'} = t^{\interp}$;
and for all $t\in\sig(O^+)$, $t^{\interp'} = t^{\interp^+}$.
By definition of domain extensibility, $\interp'\models O$ and $\interp'\models O^+$.
\hfill\qed
\end{proof}

From this theorem it follows that \Nd is sound if both the original ontology and the contextual annotation are model extensible. However, this is not a strong restriction, because the relativization of any consistent ontology is model extensible. 

\begin{theorem}[Model extensibility of relativized ontologies]\label{thm:model.extensibility.of.relativized.ontologies}
For any annotated ontology $O_\fixC = \langle O,\fixC\rangle$ where $O$ is consistent, if $\sig(O)\subseteq\nonctxtterms$ 
then $\rel_\C(O)$ is model extensible. and all its models are domain extensible.
\end{theorem}

\begin{proof}[\refthm{model.extensibility.of.relativized.ontologies}]
Since $O$ is consistent, there exists a model $\interp_o = \langle\Delta^{\interp_o},\cdot^{\interp_o}\rangle$. We define an interpretation $\interp_{\mathsf{rel}} = \langle\Delta^{\interp_o},\cdot^{\interp_{\mathsf{rel}}}\rangle$ of $\rel_\fixC(O)$ where for all $t\in\sig(O)$, $t^{\interp_\mathsf{rel}}=t^{\interp_o}$ and $\top_\fixC^{\interp_\mathsf{rel}}=\Delta^{\interp_o}$.

First, we show that $\interp_{\mathsf{rel}}$ is a model of $\rel_\fixC(O)$. Indeed, it clearly satisfies all extra axioms from \refitm{extra} in \refdef{relativization}. Second, since $\top_\fixC^{\interp_{\mathsf{rel}}}=\top^{\interp_{\mathsf{rel}}}$, replacing $\top$ with $\top_\fixC$ or adding a conjunction with $\top_\fixC$ does not change the meaning of the formulas between $\interp_o$ and $\interp_{\mathsf{rel}}$.

Second, we need to show that any extension of $\interp_\mathsf{rel}$ is still a model for $\rel_\fixC(O)$, that is, $\interp'=\langle\Delta^{\interp_o}\cup\Delta^{\interp^+},\cdot^{\interp_{\mathsf{rel}}}\rangle\models \rel_\fixC(O)$ for any set $\Delta^{\interp^+}$ such that $\Delta^o\cap\Delta^{\interp^+}=\emptyset$. Since the interpretations of the terms are the same in $\interp_\mathsf{rel}$ and $\interp'$, it follows that $\interp'$ satisfies the extra axioms in \refitm{extra}. We now need to prove that all other axioms of $\rel_\fixC(O)$ are satisfied, which we can do with a proof by induction taking advantage of the recursive definition of $\relf_\fixC$. However, we need to prove an auxiliary lemma.

\begin{lemma}\label{lem:aux}
For all concepts or roles $X$, $X^{\interp_o} = \relf_\fixC(X)^{\interp'}$.
\end{lemma}

This lemma can be proved by structural induction on the concepts and roles. Because of space restriction, we do not provide a complete proof but remark that every time a construct may lead to a different interpretation due to increasing the domain of interpretation (\eg by using the concept $\top$ or a negation), the function $\relf_\fixC$ adds a conjunction with $\top_\fixC$, so that the interpretations of the relativized concepts and roles stay confined in the original domain of interpretation.

From this lemma, it follows that if $\interp_o\models\alpha$ then $\interp'\models\rel_\fixC(\alpha)$, where $\alpha$ is any axiom with $\sig(\alpha)\subseteq\sig(O)$. Consequently, $\interp'\models\rel_\fixC(O)$.

In order to prove that all models of $\rel_\fixC(O)$ are domain extensible, we consider an arbitrary model $\interp$ of said ontology and apply the same construction that lead to $\interp'$ and use the same arguments to show that it is still a model of $\rel_\fixC(O)$.  \hfill\qed
\end{proof}

With \refthm{model.extensibility.of.relativized.ontologies} proven, we can now proceed to prove the main \refthm{soundness.of.nd}.

\begin{proof}[\refthm{soundness.of.nd}]
Let us assume that $\fixC$ is model extensible. Let $O$ be a consistent ontology such that $\sig(O)\cap\sig(\fixC)=\emptyset$ and $\sig(O)\subseteq\nonctxtterms$.

Since $O$ and $\fixC$ are consistent, there exist two models $\interp_o = \langle\Delta^{\interp_o},\cdot^{\interp_o}\rangle$ and $\langle\Delta^{\interp_c},\cdot^{\interp_c}\rangle$ of $O$ and $C$ respectively. \Wolog, we can assume that $\Delta^{\interp_o}\cap\Delta^{\interp_c}=\emptyset$.

First, from the proof of \refthm{model.extensibility.of.relativized.ontologies}, we know that from a model of $O$, we can define a model $\interp_\rel$ of $\rel_\fixC(O)$ having the same domain as $\interp_o$. Then, $\ren_\fixC(\rel_\fixC(O))$ is just a renaming of the terms in $\rel_\fixC(O)$, so a new interpretation that maps the renamed terms to the original interpretation of the original terms satisfies the renamed ontology.\footnote{Since ``truth is invariant under change of notation''~\cite[Section~2, p.7]{Goguen_Burstall:92}.}

Due to \refthm{model.extensibility.of.relativized.ontologies}, $\interp_\rel$ can be extended to include any additional elements in its domain while remaining a model of $\ren_\fixC(\rel_\fixC(O))$.

Similarly, $\ctxtfunc(\fixC)$ is like $\fixC$ with a renamed anchor so an interpretation $\interp_\rep$ that interprets $\rep(\fixC)$ as $\anchor^{\interp_c}$ and coincides with $\interp_c$ on all other terms must be a model of $\ctxtfunc_\fixC$. Moreover, we assumed that $\fixC$ is model extensible so that $\interp_\rep$ can be extended to include any additional elements.

We thus define an interpretation $\interp' = \langle\Delta^{\interp_o}\cup\Delta^{\interp_c},\cdot^{\interp'}\rangle$ that extends both $\interp_\rel$ and $\interp_\rep$ such that:
\begin{itemize}
 \item for all $t\in\sig(\fixC)$, $t^{\interp'} = t^{\interp_c}$;
 \item for all $t\in\sig(O)$, $\ren_\fixC(t)^{\interp'} = t^{\interp_o}$, and $t^{\interp'}$ is an arbitrary part of $\Delta^{\interp_o}\cup\Delta^{\interp_c}$;
 \item $\rep(\fixC)^{\interp'} = \anchor^{\interp_c}$;
 \item $\icpo^{\interp'_r}=\{\langle\ren_\fixC(t)^{\interp'_u},t^{\interp'_u}\rangle\mid t\in\sig(O)\}$;
 \item $\iic^{\interp'_r}=\{\langle\ren_\fixC(t)^{\interp'_u},\rep(\fixC)^{\interp'_u}\rangle\mid t\in\sig(O)\}$.
\end{itemize}

Let us prove that this interpretation is a model of $\cXVn_\mathsf{nd}(O_\fixC)$. 
Due to the domain extensibility of $\interp_\rel$ and $\interp_\rep$, $\interp'$ remains a model of $\ren_\fixC(\rel_\fixC(O))$ and of $\ctxtfunc(\fixC)$

Additionally, the axioms $\icpo(\ren_\fixC(t),t)$ and $\iic(\ren_\fixC(t),\rep(\fixC))$ are satisfied for all $t\in\sig(\alpha)$ by definition of $\icpo^{\interp'_r}$ and $\iic^{\interp'_r}$.\hfill\qed
\end{proof}


\subsection{Inconsistency Preservation}\label{sec:inconsistency.preservation}

In this section we prove that \Nd preserves inconsistencies. As in \refsec{soundness.of.nd}, we fix a contextual annotation $\fixC$ that has its signature in $\nonctxtterms$.

\begin{theorem}[Inconsistency preservation of \Nd]\label{thm:inconsistency.preservation.of.nd}
The contextualization function $\cXVn_\mathsf{nd}$ preserves inconsistencies. 
\end{theorem}

\begin{proof}[\refthm{inconsistency.preservation.of.nd}]
We prove the theorem by contraposition, that is, if $\cXVn(O_\fixC)$ is consistent, then $O$ is consistent.

Let $O_\fixC = \langle O,\fixC\rangle$ be an annotated ontology. Let us assume that $\cXVn(O_\fixC)$ is consistent. There exists a model $\interp = \model$ that satisfies $\ren_\fixC(\rel_\fixC(O))$. By definition, $\ren_\fixC$ is injective so there exists an inverse function $\ren_\fixC^-:\ren_\fixC(\nonctxtterms)\to\nonctxtterms$ which is itself injective, so that renaming terms $t\in\sig(\ren_\fixC(\rel_\fixC(O)))\cap\ren(\nonctxtterms)$ by $\ren_\fixC^-(t)$ gives us a model $\interp'=\langle\Delta^\interp,\cdot^{\interp'}\rangle$ of ontology $\rel_\fixC(O)$ where $\ren_\fixC^-(t)^{\interp'}=t^\interp$ for all $t\in\sig(\ren_\fixC(\rel_\fixC(O)))$.

From $\interp'$, we can define another interpretation $\interp''=\langle\top_\fixC^\interp,\cdot^{\interp'}\rangle$. For this to be a valid interpretation, all individuals must be interpreted as elements of $\top_\fixC^\interp$, all concepts must interpreted as subsets of $\top_\fixC^\interp$, and all roles must be interpreted as subsets of $\top_\fixC^\interp\times\top_\fixC^\interp$. But this is necessarily the case because the interpretation function is that of $\interp'$, and $\interp'$ is a model of $\rel_\fixC(O)$ where the axioms of \refitm{extra} in \refdef{relativization} guarantee this property.

Now, we prove that $\interp''\models \rel_\fixC(O)$. Indeed, as explained before, axioms added by \refitm{extra} are satisfied. Then, it can be checked by structural induction that for all concepts or roles $X$, $\relf_\fixC(X)^{\interp''}=\relf_\fixC(X)^{\interp'}$. Moreover, all remaining axioms in $\rel_\fixC(O)$ (other than those from $\refitm{extra}$) are the result of applying the relativization function on concepts and roles. Thus, all axioms of $\rel_\fixC(O)$ are of one of the following forms: $\relf_\fixC(X)\sqsubseteq\relf_\fixC(Y)$, $\relf_\fixC(C)(x)$, or $\relf_\fixC(R)(x,y)$, for $X, Y$ two concepts or two roles, $C$ a concept, $R$ a role, and $x, y$ two individuals. Since $\interp'$ is a model of $\rel_\fixC(O)$, if $C\sqsubseteq D\in\ax(O)$, then $\interp'\models\relf_\fixC(C)\sqsubseteq\relf_\fixC(D)$. This, combined with the equality of $\relf_\fixC(X)^{\interp''}$ and $\relf_\fixC(X)^{\interp'}$ for any concept or role $X$, ensures that $\interp''\models\relf_\fixC(C)\sqsubseteq\relf_\fixC(D)$. The same line of reasoning holds for $C(x)\in\ax(O)$ and $R(x,y)\in\ax(O)$. It follows that $\interp''\models\rel_\fixC(O)$.

We finish the proof by showing that $\interp''\models O$, which proves that $O$ is consistent. Since $\top^{\interp''}=\top_\fixC^{\interp''}$, adding $\sqcap\top_\fixC$ to a concept does not change its interpretation, and replacing $\top$ with $\top_\fixC$ has no effect on the interpretation of the concepts or roles. Consequently, for all concepts or roles $X$, $X^{\interp''}=\relf_\fixC(X)^{\interp''}$. Thus, when $\interp''\models\rel_\fixC(\alpha)$ it also satisfies $\alpha$ and therefore, $\interp''\models O$.\hfill\qed
\end{proof}


\subsection{Inference Preservation}\label{sec:inference.preservation}

In~\cite{GimenezGarcia_Zimmermann_Maret:17}, we were only able to study entailment preservation in the limited setting of $pD*$ entailment. Here we prove a much stronger theorem for \Nd. As in the two previous subsections, we fix a contextual annotation $\fixC$ that has its signature in $\nonctxtterms$.

\begin{theorem}[Entailment preservation of \Nd]\label{thm:entailment.preservation.of.nd}
Let $\Omega$ be a set of ontologies having their signatures in $\nonctxtterms$ and disjoint from the signature of $\fixC$. If $\fixC$ is model extensible, then \Nd is entailment preserving for $\{\langle O,\fixC\rangle\}_{O\in\Omega}$.
\end{theorem}

In order to prove this theorem, we must show that relativization preserves entailments.

\begin{lemma}\label{lem:preservation.of.entailments.in.rel}
Let $O_1$ and $O_2$ be two ontologies. If $O_1\models O_2$ then $\rel_\fixC(O_1)\models\rel_\fixC(O_2)$.
\end{lemma}

\begin{proof}[\reflem{preservation.of.entailments.in.rel}]
Let $O_1$ and $O_2$ such that $O_1\models O_2$. If $\rel_\fixC(O_1)$ is inconsistent, then the property is obviously verified. Otherwise, there exists a model $\interp = \model$ of $\rel_\fixC(O_1)$. Using the same arguments as in the proof of \refthm{inconsistency.preservation.of.nd}, $\interp'=\langle\top_\fixC,\cdot^\interp\rangle$ is also a model of $\rel_\fixC(O_1)$ and $\interp'$ interprets concepts and roles equally to their relativized counterparts, \ie $X^{\interp'}=\relf_\fixC(X)^{\interp'}$. It follows that $\interp'\models O_1$ and since $O_1\models O_2$, $\interp'\models O_2$, and $\interp'\models \rel_\fixC(O_2)$. By \refthm{model.extensibility.of.relativized.ontologies}, we know that $\interp'$ is domain extensible for $\rel_\fixC(O_2)$. Consequently, $\interp$ is also a model of $\rel_\fixC(O_2)$.\hfill\qed
\end{proof}

With the lemma proven, we then prove \refthm{entailment.preservation.of.nd}.

\begin{proof}[\refthm{entailment.preservation.of.nd}]
Let us assume that $\fixC$ is model extensible. Let $O_1$ and $O_2$ be two ontologies having their signatures in $\nonctxtterms$ and disjoint from $\sig(\fixC)$, such that $O_1\models O_2$.
If $\cXVn(\langle O_1,\fixC\rangle)$ is inconsistent, then the entailment is trivially preserved. Otherwise, there exists a model $\interp = \model$ of $\cXVn(\langle O_1,\fixC\rangle)$. By definition of \cXVn, it satisfies $\ren_\fixC(\rel_\fixC(O_1)$. Due to \reflem{preservation.of.entailments.in.rel}, $O_1\models O_2$ implies that $\rel_\fixC(O_1) \models \rel_\fixC(O_2)$. Moreover, using the fact that $\ren_\fixC$ is just a renaming of terms, and considering that ``truth is invariant under change of notation''~\cite{Goguen_Burstall:92}, we know that $\ren_\fixC(O)\models\ren_\fixC(O')$ is equivalent to $O\models O'$. Therefore, $\ren_\fixC(\rel_\fixC(O_1) \models \ren_\fixC(\rel_\fixC(O_2)$. Moreover, all axioms in $\ctxtfunc(\fixC)$ are satisfied by $\interp$. Finally, for $O_2$ to be entailed by $O_1$, the signature of $O_2$ must be included in the signature of $O_1$. As a result, whenever $\icpo(\ren_\fixC,t)$ or $\iic(\ren_\fixC(t),\rep(\fixC))$ are in $\cXVn(\langle O_2,\fixC\rangle)$, then they are in $\cXVn(\langle O_1,\fixC\rangle)$ as well. Consequently, $\cXVn(\langle O_1,\fixC\rangle)\models\cXVn(\langle O_2,\fixC\rangle)$.\hfill\qed
\end{proof}


\section{Annotations in Multiple Contexts} \label{sec:multiple.domains}

So far, we assumed that all axioms of an ontology are annotated with the same contextual information. In this setting, the core of the contextualization function in \Nd amounts to relativizing the axioms and renaming the terms. The renaming part may seem surprising because, as we said a couple of times already, ``truth is invariant under change of notation''. However, the usefulness of the renaming part becomes apparent when we want to combine several annotated ontologies having different contextual annotations (say $\C_1$ and $\C_2$). In this case, if the renaming functions $\ren_{\C_1}$ and $\ren_{\C_2}$ are mapping non contextual terms into disjoint sets of contextual terms, then the contextualization function $\cXVn_\mathsf{nd}$ ensures that any inference made in a context will not interact with the knowledge from another context. This avoids the contextualized knowledge to be inconsistent when combining statements in different contexts that contradict each others.


The properties presented in \refsec{context} require a little adaptation when applied to the multi-contextual setting. Indeed, in spite of the soundness theorem of \refsec{soundness.of.nd}, in the general case if a set of annotated ontologies $\{\langle O_i,\C_i\rangle\}_{i\in I}$ are satisfying the constraints of \refthm{soundness.of.nd}, it is still possible that $\bigcup_{i\in I}\cXVn_\mathsf{nd}(\langle O_i,\C_i\rangle)$ is inconsistent. We expect that the preservation of consistency can be guaranteed if \emph{all} the signatures of $\{O_i\}$ are disjoint from \emph{all} the signatures of $\{\C_i\}$. Studying in more details the case of multiple contextual annotations is planned for future work.

\section{Other Approaches}\label{sec:other.approaches}

Here we briefly present the most relevant reification approaches in the Semantic Web. For all of them, the contextualization only annotates the role assertions, leaving other axioms unmodified. 

As seen in \refex{contextualization}, \emph{RDF reification} replaces $\alpha = R(x,y)$ with three new role assertions $\ct{subject}(\anchor_\alpha^\C,x)$, $\ct{predicate}(\anchor_\alpha^\C,R)$, and $\ct{object}(\anchor_\alpha^\C,y)$ and the axioms in the contextual annotation are anchored on $\anchor_\alpha^\C$, which depends on the role assertion $R$ and the contextual annotation $\C$. As shown in \refex{inconsistency.preservation} and \ref{ex:entailment.preservation}, this contextualization method preserves neither inconsistencies (in the sense of \refdef{inconsistency.preservation}) nor entailments on role assertions.

\begin{example}
An \emph{RDF reification} contextualization of our running example within the context $\Cex$ of \refex{contextual.annotation} contains the following axioms:

$\ct{subject}(\ct{stbcobe},\ct{babylon})$

$\ct{predicate}(\ct{stbcobe},\ct{capital})$

$\ct{object}(\ct{stbcobe},\ct{babylonianEmpire})$

$\ct{validity}(\ct{stbcobe},\ct{t})$

$\ct{Interval}(\ct{t})$

$\ct{from}(\ct{t},\ct{609BC})$

$\ct{to}(\ct{t},\ct{539BC})$

$\ct{prov}(\ct{stbcobe,\ct{w}})$

$\ct{name}(\ct{w},\ct{wikipedia})$

$\ct{Wiki}(\ct{w})$
\end{example}

\emph{N-Ary relations} replaces $R(x,y)$ by two role assertions $p_1(R)(x,\anchor_\alpha^\C)$ and $p_2(R)(\anchor_\alpha^\C,y)$, where $R$ is a simple role assertion, and $p_1$ and $p_2$ are two injective functions with disjoint ranges that map non-contextual roles to contextual roles. Alternatively, a new concept $C_R$ is added for the role $R$, and the following assertions are added: $C_R(\anchor_\alpha^\C)$, $p_1(R)(\anchor_\alpha^\C,x)$, and $p_2(R)(\anchor_\alpha^\C,y)$.

\begin{example}
An \emph{N-Ary relations} contextualization of our running example within the context $\Cex$ of \refex{contextual.annotation} contains the following axioms:

$\ct{capitalOf\#1}(\ct{babylon},\ct{rbcobe})$

$\ct{capitalOf\#2}(\ct{rbcobe},\ct{babylonianEmpire})$

$\ct{validity}(\ct{rbcobe},\ct{t})$

$\ct{Interval}(\ct{t})$

$\ct{from}(\ct{t},\ct{609BC})$

$\ct{to}(\ct{t},\ct{539BC})$

$\ct{prov}(\ct{rbcobe,\ct{w}}$

$\ct{name}(\ct{w},\ct{wikipedia})$

$\ct{Wiki}(\ct{w})$
\end{example}

\emph{Singleton property} is using a non-standard semantics of RDF but the same idea can be simulated with DL axioms. For each simple role axiom $R(x,y)$, the following axioms are added: $\anchor_\alpha^\C(x,y)$ (that is, the term for the anchor is used as a role), $\{x\}\equiv\exists\anchor_\alpha^\C.\{y\}$ (which guarantees that the anchor property is a singleton), and $\ct{singletonPropertyOf}(\anchor_\alpha^\C,R)$.

\begin{example}
A \emph{Singleton Property} contextualization of our running example within the context $\Cex$ of \refex{contextual.annotation} contains the following axioms:

$\ct{capital\#1}(\ct{babylon},\ct{babylonianEmpire})$

$\ct{singletonPropertyOf}(\ct{capital\#1},\ct{capital})$

$\ct{validity}(\ct{capital\#1},\ct{t})$

$\ct{Interval}(\ct{t})$

$\ct{from}(\ct{t},\ct{609BC})$

$\ct{to}(\ct{t},\ct{539BC})$

$\ct{prov}(\ct{capital\#1,\ct{w}}$

$\ct{name}(\ct{w},\ct{wikipedia})$

$\ct{Wiki}(\ct{w})$
\end{example}

The remaining approach, \emph{NdFluents}, uses a similar approach as \Nd except that it only renames the terms used as individuals and does not relativize the ontology. This ensures interesting properties \wrt entailment preservation~\cite{GimenezGarcia_Zimmermann_Maret:17}, but TBox axioms in different contexts are not distinguishable.

\section{Discussion and Future Work} \label{sec:discussion} 


\Nd and NdFluents are a concrete instantiations of a general approach of contextualizing (parts of) the terms in the ontology. Other instantiations would be possible, such as contextualizing role names (in a similar fashion as the singleton property), class names, or a combination of them. Then, \Nd would be the approach where each and every term is contextualized, while in NdFluents only individuals are.

In the future, we would like to deepen the analysis of contextualization, filling gaps still present in this preliminary work. Especially, the combination of multiple annotations, or annotations of contextualized ontologies, present some interesting challenges. A more systematic comparison of the various approaches remains to be presented.

\paragraph{Acknowledgement:} This work was partly funded by project WDAqua (H2020-MSCA-ITN-2014 \#64279) and project OpenSensingCity (ANR-14-CE24-0029).

\bibliographystyle{splncs}
\bibliography{biblioaz}

\end{document}